\documentclass[oneside,english]{amsart}
\usepackage[OT1]{fontenc}
\usepackage[latin9]{inputenc}
\usepackage{textcomp}
\usepackage{amsthm}

\makeatletter
\numberwithin{equation}{section}
\numberwithin{figure}{section}
\theoremstyle{plain}
\newtheorem{thm}{Theorem}
  \theoremstyle{definition}
  \newtheorem{defn}[thm]{Definition}
  \theoremstyle{plain}
  \newtheorem{cor}[thm]{Corollary}
  \theoremstyle{plain}
  \newtheorem{lem}[thm]{Lemma}
  \theoremstyle{plain}
  \newtheorem{prop}[thm]{Proposition}

\usepackage{fullpage}

\makeatother

\usepackage{babel}

\begin{document}

\title{Boson Sampling is Robust to Small Errors in the Network Matrix}

\author{Alex Arkhipov}

\thanks{This research was supported by Scott Aaronson's NSF Waterman Award.}
\begin{abstract}
We demonstrate the robustness of BosonSampling to imperfections in
the linear optical network that cause a small deviation in the matrix
it implements. We show that applying a noisy matrix $\tilde{U}$ that
is within $\epsilon$ of the desired matrix $U$ in operator norm
leads to an output distribution that is within $\epsilon n$ of the
desired distribution in variation distance, where $n$ is the number
of photons. This lets us derive a sufficient tolerance each beamsplitters
and phaseshifters in the network. 

This result considers only errors that result from the network encoding
a different unitary than desired, and not other sources of noise such
as photon loss and partial distinguishability.
\end{abstract}
\maketitle

\section{Background}

\subsection{BosonSampling}

BosonSampling \cite{key-1} is a computational problem inspired by
linear optics and closely related to the matrix permanent. It models
a one-step linear-optical experiment where $n$ identical photons
are produced in distinct modes and passed simultaneously through a
linear network that encodes an $m\times m$ unitary matrix $U$. The
initial state $\left|1_{n}\right\rangle $ consists of one boson in
each of the first $n$ modes, with the rest empty%
\footnote{In this work, we will loosen this assumption to allow any pure $n$-boson
state%
}. A photon counting measurement is performed on each output mode,
and we consider the outcome to be the list of photon counts $S=\left(s_{1},\dots s_{m}\right)$,
where the $s_{i}$ are nonnegative integers whose sum is $S$. This
output is random, and we define $\mathcal{D}_{U}$ to be the resulting
probability distribution over outcomes.%
\footnote{We deviate slightly from the definitions in \cite{key-1} in which
considers only the $m\times n$ submatrix $A$ of $U$ consisting
of the first $n$ rows, the ones relevant to that start state $\left|1_{n}\right\rangle $,
and defines the distribution $\mathcal{D}_{A}$ in terms of this.%
}

As a computational problem, we can define BosonSampling as follows: 
\begin{defn}
[BosonSampling]\label{def:Boson sampling dist}Given an $m\times m$
matrix $U$ and a parameter $n$, sample the distribution of photon
counts $\mathcal{D}_{U}$ given by \[
\Pr_{D_{U}}\left[S\right]=\frac{\left|\mathrm{Per}\left(U_{\left[n\right],S}\right)\right|^{2}}{s_{1}!\cdots s_{m}!},\]
where Per is the matrix permanent and $U_{\left[n\right],S}$ is the
submatrix of $U$ consisting of the first $n$ rows and the columns
given by $S$ with multiplicity.
\end{defn}
We can think of the linear optical network as acting on the $n$-photon
Hilbert space each of whose basis elements is labelled by each photon
count. Its dimension is $N=\left(\binom{m}{n}\right)$, the number
of partitions of $n$ unlabelled photons into $m$ labelled modes.
Let $\varphi$ be the homomorphism from the unitary $m\times m$ unitary
$U$ defining the action of one photon the $N\times N$ unitary matrix
$\varphi\left(U\right)$ defining the action on $n$ identical photons.
See \cite{key-1} for a precise definition of $\varphi$ and a proof
that is is indeed a homomorphism.

The definition of BosonSampling is motivated by a surprising result
about its computational complexity that gives evidence that it cannot
be approximated by a classical computer, and thus evidence against
the Church-Turing Thesis.

\subsection{Experimental realization}

The definition of BosonSampling is partially motivated by it modelling
a linear optical experiment. Moreover, the prospect of a computation
beyond what is possible in the efficient classical world asks for
a such a computational device to be built and tested, if only to check
that quantum mechanics works as we'd expect. 

Four independent labs (based in Brisbane \cite{Brisbane}, Oxford
\cite{Oxford}, Vienna \cite{Vienna}, and Rome \cite{Rome}) have
built devices to implement the BosonSampling setup for small numbers
of photons and modes, and checked the results to be as statistically
expected. These experiments were done with $n=3$ photons are a number
of modes $m$ ranging between $5$ and $9$. While these are modest
paramaters, the goal is to eventually scale up the experiments to
the point that the problem it solves is intractable for the fastest
classical computers of the time.

\subsection{Experimental noise}

The question of scaling naturally leads to the issue of noise. Real
experiments have imperfections that cause them to slighly deviate
from the ideal model, and we would like to understand what level of
error is tolerable in order to get only a small deviation in the output
distribution.

There are four main sources of noise to consider:
\begin{enumerate}
\item Incorrect or correlated initial states
\item Imperfect coding of the unitary $U$ by the linear optical network
\item Partial distinguishability of photons (caused by non-simultaneous
arrival)
\item Photon loss (whether in the network or due to failure to measure)
\end{enumerate}
In this work, we solely consider (2), the effect caused by imperfections
in the linear optical network that cause a deviation in the unitary
matrix that it encodes. We assume that the actual network still applies
a unitary matrix $\tilde{U}$ (in particular, it takes pure states
to pure states), but one that is slightly different from the desired
matrix $U$. In current experiments, this noise source is not a limiting
significant factor -- beamsplitters and phaseshifters can be produced
rather accurately. However, we investigate it so as to understand
its effect in the asymptotic regime, which will become more and more
relevant as the experiments scale up.

We give an upper bound for the error in the output distribution in
terms of the error in $U$. In particular, we show that $o\left(1/n\right)$
operator distance suffices to give $o\left(1\right)$ error in the
output distribution.

\section{Main Result}

Our main result is a bound on the error in the BosonSampling distribution
$D_{U}$ (Definition \ref{def:Boson sampling dist}) caused by inaccuracy
in the single-particle unitary $U$ that encodes the action of the
beamsplitters and phaseshifters.
\begin{thm}
For unitary matrices $U$ and $\tilde{U}$, the $L_{1}$ distance
between the corresponding $n$-photon BosonSampling distributions
$\mathcal{D}_{U}$ and $\mathcal{D}_{\tilde{U}}$ is bounded as\[
\left\Vert \mathcal{D}_{\tilde{U}}-\mathcal{D}_{U}\right\Vert _{1}\leq n\left\Vert \tilde{U}-U\right\Vert _{\mathrm{op}}\]

\end{thm}
Note that there in no dependence on the number of modes $m$. As a
result, the accuracy of the unitaries only needs to depend on the
number of photons $n$, with $o\left(\frac{1}{n}\right)$ error sufficing.
\begin{cor}
To obtain vanishingly small error \textup{$\left\Vert \mathcal{D}_{\tilde{U}}-\mathcal{D}_{U}\right\Vert _{1}=o\left(1\right)$},
it suffices for the unitary representing the entire transformation
to have $\left\Vert \tilde{U}-U\right\Vert _{\mathrm{op}}=o\left(\frac{1}{n}\right)$.
\end{cor}
This can be achieved by having each beamsplitter and phaseshifter
in the network be sufficiently accurate. Since such a network can
be made with a depth of $O\left(n\log m\right)$ components (Theorem
45 of \cite{key-1}), it suffices to divide the tolerable error by
that amount.
\begin{cor}
In order to have $\left\Vert \mathcal{D}_{\tilde{U}}-\mathcal{D}_{U}\right\Vert _{1}=o\left(1\right)$,
it suffices for every component in the network to have an accuracy
of $\left\Vert \tilde{A}-A\right\Vert _{\mathrm{op}}=o\left(\frac{1}{n^{2}\log m}\right)$.
\end{cor}

\section{Comparison to Previous Results}

\subsection{Comparison to previous work}

The result is comparable to the standard result for qubit-based circuits
of Bernstein and Vazirani \cite{BV}, stated here with identical gates
and in language to parallel our main result.
\begin{thm}
Suppose one applies a noisy unitary matrix $\tilde{U}$ to each of
$n$ distinguishable particles, then measures the outcome distribution.
Then, the distance in the outcome distribution $\mathcal{D}_{\tilde{U}}$
from having done so with matrix $U$ is bounded as 

\[
\left\Vert \mathcal{D}_{\tilde{U}}-\mathcal{D}_{U}\right\Vert _{1}\leq n\left\Vert \tilde{U}-U\right\Vert _{\mathrm{op}}\]

\end{thm}
Previous work on BosonSampling noise sensitivity has given necesarry
bounds for the required accuracy of the linear optical network. In
other words, it's shown that above certain thresholds of noise, one
gets large inaccuracies in the distribution of outcomes. Thus, it
proves a certain level of noise to be prohibitive for BosonSampling.
This work, in contrast, shows a certain level of accuracy to be sufficient.

The work of Leverrier and Patrón \cite{Leverrier Patron} demonstrates
that each linear optical element must have fidelity $1-O\left(1/n^{2}\right)$
by considering a composite experiment in which the network is applied
followed by its inverse, with independent noise in each part. As shown
in Appendix A, this corresponds to a required single-operator distance
of $O\left(1/n^{2}\right)$, which has a factor of $\log m$ gap from
our sufficient bound of $o\left(1/\left(n^{2}\log m\right)\right)$
per operator being sufficient. From our methods in Section \ref{sub:Components},
this implies an overall distance of $\left\Vert \tilde{U}-U\right\Vert _{op}=O\left(\log m/n\right)$,
again a factor of $\log m$ off of our result.

The work of Kalai and Kindler \cite{Kalai Kindler} argues that a
noise level of additive $\omega\left(1/n\right)$ Gaussian error applied
to the overall unitary matrix leads to large deviations in the output
distribution. Specifically, above such a threshold, one finds vanishingly
little correlation between the original and noise permanent of a submatrix,
and thus between outcomes of a BosonSampling experiment. Translating
to our error model of unitary noise in Appendix A, a typical such
error corresponds to operator distance $\omega\left(1/\sqrt{n}\right)$,
significantly above the $O\left(1/n\right)$ distance that we show.

In both cases, once we convert the error measures to a consistent
scale, we find the sufficient bound for noise shown in this work is
consistent with the necessarry bound shown in the previous. Moreover,
a gap remains for potential improvement.

\section{Proof of Result}

\subsection{Outline of proof}

We give an outline of the proof here, and prove each part in the upcoming
sections.

Let $\Psi_{0}$ be the initial $n$-boson state, and let $\varphi$
be the homomorphism from a unitary acting on one boson to that acting
on $n$ identical bosons. Applying unitaries $U$ and $\tilde{U}$
respectively to the initial state $\Psi_{0}$ produce: \begin{eqnarray*}
\Psi & = & \varphi\left(U\right)\Psi_{0}\\
\tilde{\Psi} & = & \varphi\left(\tilde{U}\right)\Psi_{0}\end{eqnarray*}
Measuring $\Psi$ and $\tilde{\Psi}$ respectively in the standard
basis gives outcome distributions $\mathcal{D}_{U}$ and $\mathcal{D}_{\tilde{U}}$ 

The main step is Theorem \ref{thm:homomorphism expansion}, which
states that the distance between the $n$-boson unitaries is at most
a factor of $n$ times that between the $1$-boson unitaries 

\[
\left\Vert \varphi\left(\tilde{U}\right)-\varphi\left(U\right)\right\Vert _{\mathrm{op}}\leq n\left\Vert \tilde{U}-U\right\Vert _{\mathrm{op}}\]
We then conclude with a standard argument (Lemma \ref{lem:Operator distance effect})
that the distance between the output distributions is at most the
operator distance between the matrices that produced them 

\[
\left\Vert \mathcal{D}_{\tilde{U}}-\mathcal{D}_{U}\right\Vert _{1}\leq\left\Vert \varphi\left(\tilde{U}\right)-\varphi\left(U\right)\right\Vert _{\mathrm{op}}\]

\subsection{Effect of the homomorphism\label{sec:Effects-of-the-homomorphism}}

We first show that close unitaries $U$ and $\tilde{U}$ induce nearby
$n$-boson unitaries $\varphi\left(U\right)$ and $\varphi\left(\tilde{U}\right)$.
Thus, if two operations act similarly on single bosons, then they
also act similarly on $n$ identical bosons. The blowup is simply
a factor of $n$, the number of bosons.
\begin{thm}
\label{thm:homomorphism expansion}Let $\varphi$ be the homomorphism
that takes a $m\times m$ unitary matrix $U$ acting on a single boson
and produces a $N\times N$ unitary matrix acting on $n$ identical
bosons with $N=\left(\binom{m}{n}\right)$. Then, \[
\left\Vert \varphi\left(\tilde{U}\right)-\varphi\left(U\right)\right\Vert _{\mathrm{op}}\leq n\left\Vert \tilde{U}-U\right\Vert _{\mathrm{op}}\]

\end{thm}
In order to prove this, it will be useful to have two lemmas. Lemma
\ref{lem:Operator distance in terms of eigenvalues} expresses the
operator distance between two unitary matrices $A$ and $B$ in terms
of the eigenvalues of $AB^{-1}$. Lemma \ref{lem:Eigenvalues of homomorphism}
relates the eigenvalues of $\varphi\left(M\right)$ to those of $M$.
\begin{lem}
\label{lem:Operator distance in terms of eigenvalues} If $A$ and
$B$ are unitary, their operator distance can be expressed in terms
of the eigenvalues $\left\{ \lambda_{i}\right\} $ of $AB^{-1}$ as
\[
\left\Vert A-B\right\Vert _{\mathrm{op}}=\max_{i}\left|\lambda_{i}-1\right|.\]
\end{lem}
\begin{proof}
Since $AB^{-1}$ is unitary, it diagonalizes via unitaries as $AB^{-1}=V\mbox{diag\ensuremath{\left(\lambda_{i}\right)}\ensuremath{\ensuremath{V^{^{\ast}}}}}$.
Use the operator norm's invariance to left-multiplication or right-multiplication
by a unitary, we have

\begin{eqnarray*}
\left\Vert A-B\right\Vert _{\mathrm{op}} & = & \left\Vert AB^{-1}-I\right\Vert _{op}\\
 & = & \left\Vert V\left(\mathrm{diag}\left(\lambda_{i}\right)-I\right)V^{\ast}\right\Vert _{\mathrm{op}}\\
 & = & \left\Vert \mathrm{diag}\left(\lambda_{i}-1\right)\right\Vert _{\mathrm{op}}\\
 & = & \max_{i}\left|\lambda_{i}-1\right|.\end{eqnarray*}
\end{proof}
\begin{lem}
\label{lem:Eigenvalues of homomorphism} If $M$ has eigenvalues $\left(\lambda_{1},\dots,\lambda_{m}\right)$,
then the eigenvalues of $\varphi\left(M\right)$ are $\lambda_{1}^{s_{1}}\cdots\lambda_{m}^{s_{m}}$
for each ordered partition $S$ of $n$ into $m$ parts with sizes
$s_{1},\dots,s_{m}$.\end{lem}
\begin{proof}
Let $v_{i}$ be the eigenvector corresponding to $\lambda_{i}$. We
will construct eigenvectors of $\varphi\left(M\right)$ in terms of
the $v_{i}$ and note that they have the desired eigenvalues.

For each eigenvector $v_{i}$, let $v_{i}\left(x\right)$ be the formal
polynomial $\left(v_{i}\right)_{1}x_{1}+\cdots+\left(v_{i}\right)_{n}x_{n}$.
For each $S$, let $p_{S}$ be the degree-$n$ polynomial\[
p_{S}\left(x\right)=v_{1}^{s_{1}}\left(x\right)\dots\cdots v_{m}^{s_{m}}\left(x\right)\]
If we consider $\varphi\left(M\right)$ as it acts on the Fock basis,
we see that each $p_{S}\left(x\right)$ is an eigenvector with eigenvalue
$\lambda_{1}^{s_{1}}\cdots\lambda_{m}^{s_{m}}$: 

\begin{eqnarray*}
\varphi\left(M\right)\left(p_{S}\left(x\right)\right) & = & \left(Mv_{1}\right)^{s_{1}}\left(x\right)\cdots\left(Mv_{m}\right)^{s_{m}}\left(x\right)\\
 & = & \left(\lambda_{1}v\right)^{s_{1}}\left(x\right)\cdots\left(M\lambda_{m}v\right)^{s_{m}}\left(x\right)\\
 & = & \lambda_{1}^{s_{1}}\cdots\lambda_{m}^{s_{m}}\left(p_{S}\left(x\right)\right)\end{eqnarray*}
Since we have one eigenvalues for each $S$, the number of which equal
the dimension $\left(\binom{m}{n}\right)$ of $\varphi\left(M\right)$,
this is the full set of eigenvalues.
\end{proof}
Now, we're ready to prove Theorem \ref{thm:homomorphism expansion},
which we restate here.
\begin{thm}
Let $\varphi$ be the homomorphism that takes a $m\times m$ unitary
matrix $U$ acting on a single boson and produces a $N\times N$ unitary
matrix acting on $n$ identical bosons with $N=\left(\binom{m}{n}\right)$.
Then, \[
\left\Vert \varphi\left(\tilde{U}\right)-\varphi\left(U\right)\right\Vert _{\mathrm{op}}\leq n\left\Vert \tilde{U}-U\right\Vert _{\mathrm{op}}\]
\end{thm}
\begin{proof}
Let $\left(\lambda_{1},\dots,\lambda_{m}\right)$ be the eigenvalues
of $\tilde{U}U^{-1}$. From Lemma \ref{lem:Eigenvalues of homomorphism},
the eigenvalues of $\varphi\left(\tilde{U}\right)\varphi\left(U\right)^{-1}$,
which equals $\varphi\left(\tilde{U}U^{-1}\right)$ because $\varphi$
is a homomorphism, are $\lambda_{1}^{s_{1}}\cdots\lambda_{m}^{s_{m}}$
for each ordered partition $S$ of $n$ into $m$ parts, which we
write as $\lambda^{S}$ for brevity. 

We now bound the distance of $\lambda^{S}$ from $1$ in terms of
the distances of the $\lambda_{i}$ from $1$. As eigenvalues of a
unitary matrix, the $\lambda_{i}$ are complex phases with norm $1$,
we can inductively apply \begin{eqnarray*}
\left|ab-1\right| & = & \left|ab-a+a-1\right|\\
 & \leq & \left|a\right|\left|b-1\right|+\left|a-1\right|\end{eqnarray*}
to get

\begin{equation}
\left|\lambda^{S}-1\right|\leq\sum_{i}s_{i}\left|\lambda_{i}-1\right|\leq n\max_{i}\left|\lambda_{i}-1\right|\label{eq:eigenvalues distance from 1 bound}\end{equation}
From Lemma \ref{lem:Operator distance in terms of eigenvalues}, we
have \[
\max_{i}\left|\lambda_{i}-1\right|=\left\Vert \tilde{U}-U\right\Vert _{\mathrm{op}}\]
and \[
\max_{S}\left|\lambda^{S}-1\right|=\left\Vert \varphi\left(\tilde{U}\right)-\varphi\left(U\right)\right\Vert _{\mathrm{op}},\]
so equation \ref{eq:eigenvalues distance from 1 bound} gives the
desired result \[
\left\Vert \varphi\left(\tilde{U}\right)-\varphi\left(U\right)\right\Vert _{\mathrm{op}}\leq n\left\Vert \tilde{U}-U\right\Vert _{\mathrm{op}}\]

\end{proof}

\subsection{Bounding distance between the output distributions}

In Section \ref{sec:Effects-of-the-homomorphism}, we showed that
$\tilde{U}$ being close to $U$ implies that the corresponding $n$-boson
transition matrices $\varphi\left(U\right)$ and $\varphi\left(\tilde{U}\right)$
are close. We now argue that applying close transition matrices to
the same input produces close measurement distributions.

Let $\Psi_{0}$ be the initial $n$-boson state. For BosonSampling,
this is a Fock basis state $\left|1_{n}\right\rangle $, but this
is not necesarry for this result. Applying unitaries $U$ and $\tilde{U}$
to $\Psi_{0}$ produces states that we call\begin{eqnarray*}
\Psi & = & \varphi\left(U\right)\Psi_{0}\\
\tilde{\Psi} & = & \varphi\left(\tilde{U}\right)\Psi_{0}\end{eqnarray*}
The distributions $\mathcal{D}_{U}$ and $\mathcal{D}_{\tilde{U}}$
are produced by measuring $\Psi$ and $\tilde{\Psi}$ respectively
in the standard basis.

We show that the distance between the distributions is bounded by
the operator distance between the respective operators that produced
them.
\begin{lem}
\label{lem:Operator distance effect}\textup{$\left\Vert \mathcal{D}_{\tilde{U}}-\mathcal{D}_{U}\right\Vert \leq\left\Vert \varphi\left(\tilde{U}\right)-\varphi\left(U\right)\right\Vert _{\mathrm{op}}$}\end{lem}
\begin{proof}
We first bound the Euclidian distance of the resulting states from
the definition of the operator norm \begin{eqnarray}
\left\Vert \tilde{\Psi}-\Psi\right\Vert  & = & \left\Vert \left(\varphi\left(U\right)-\varphi\left(\tilde{U}\right)\right)\Psi_{0}\right\Vert \nonumber \\
 & \leq & \left\Vert \varphi\left(\tilde{U}\right)-\varphi\left(U\right)\right\Vert _{\mathrm{op}}\left\Vert \Psi_{0}\right\Vert \nonumber \\
 & = & \left\Vert \varphi\left(\tilde{U}\right)-\varphi\left(U\right)\right\Vert _{\mathrm{op}}\label{eq:state-dist-norm}\end{eqnarray}
Now, we show that variation distance between $\mathcal{D}_{U}$ and
$\mathcal{D}_{\tilde{U}}$ is bounded by this distance$\left\Vert \tilde{\Psi}-\Psi\right\Vert $. 

The variation distance $\left\Vert \mathcal{D}_{\tilde{U}}-\mathcal{D}_{U}\right\Vert _{1}$
corresponding to the distributions obtained from a standard basis
measurement is bounded by the trace distance, the maximum such variation
over all projective measurements.\[
\left\Vert \mathcal{D}_{\tilde{U}}-\mathcal{D}_{U}\right\Vert _{1}\leq\left\Vert \tilde{\Psi}-\Psi\right\Vert _{\mathrm{tr}}\]
We use the expression for trace distance between pure states and bound
this expression in terms of $\left\Vert \tilde{\Psi}-\Psi\right\Vert .$\begin{eqnarray*}
\left\Vert \mathcal{D}_{\tilde{U}}-\mathcal{D}_{U}\right\Vert _{1} & \leq & \left\Vert \tilde{\Psi}-\Psi\right\Vert _{\mathrm{tr}}\\
 & = & \sqrt{1-\left|\left\langle \tilde{\Psi}\mid\Psi\right\rangle \right|^{2}}\\
 & \leq & \sqrt{1-\left(\mathrm{Re}\left\langle \tilde{\Psi}\mid\Psi\right\rangle \right)^{2}}\\
 & = & \sqrt{1-\left(1-\frac{1}{2}\left\Vert \tilde{\Psi}-\Psi\right\Vert \right)^{2}}\\
 & \leq & \left\Vert \tilde{\Psi}-\Psi\right\Vert \end{eqnarray*}

\end{proof}
Combining this with Equation \ref{eq:state-dist-norm} gives the bound
\[
\left\Vert \mathcal{D}_{\tilde{U}}-\mathcal{D}_{U}\right\Vert _{1}\leq\left\Vert \varphi\left(\tilde{U}\right)-\varphi\left(U\right)\right\Vert _{\mathrm{op}},\]
which, along with Theorem \ref{thm:homomorphism expansion} \[
\left\Vert \varphi\left(\tilde{U}\right)-\varphi\left(U\right)\right\Vert _{\mathrm{op}}\leq n\left\Vert \tilde{U}-U\right\Vert _{\mathrm{op}}\]
gives the main result.

\subsection{Error tolerance of components of the linear optical network\label{sub:Components}}

We now investigate the maximum error on components of the linear optical
network that still guarantees that the a output distribution is vanishingly
close to the ideal one. This requires bounding the error of the unitary
produced by a linear optical network in terms of that of its components.
\begin{prop}
If each component $\tilde{A}$ of a linear optical network is within
operator distance $\epsilon$ of the ideal component $A$ \[
\left\Vert \tilde{A}-A\right\Vert _{\mathrm{op}}\leq\epsilon,\]
then the produced unitary $U$ acting on the first $n$ modes has
accuracy

\[
\left\Vert \tilde{U}-U\right\Vert _{\mathrm{op}}=O\left(n\epsilon\log m\right)\]
and the measured output has
\end{prop}
\[
\left\Vert \mathcal{D}_{\tilde{U}}-\mathcal{D}_{U}\right\Vert _{1}=O\left(n^{2}\epsilon\log m\right).\]

\begin{proof}
We wish to bound the operator distance error of the network in terms
of that of its components. We use two familiar facts about operator
distance: 
\begin{itemize}
\item For components are applied in parallel, the overall operator distance
error is at most that of each component, So, if each component has
some maximum error, so does each layer in the network. 
\item For components applied in series, the total operator distance error
is at most the sum of the operator distance error of the components.
\end{itemize}
A linear optical network for $n$ fixed input modes and $m$ output
modes can be implemented using $O\left(mn\right)$ beamsplitters and
phaseshifters in a network of depth $O\left(n\log m\right)$ (Theorem
45 of \cite{key-1}). So, if each optimal element is within operator
norm $\epsilon$ of the ideal, we are guaranteed the following accuracy
for a linear optical network

\[
\left\Vert \tilde{U}-U\right\Vert _{op}=O\left(n\log m\right)\left\Vert \tilde{A}-A\right\Vert _{\mathrm{op}}=O\left(n\epsilon\log m\right)\]
Applying the main theorem then gives an overall error of 
\end{proof}
\[
\left\Vert \mathcal{D}_{\tilde{U}}-\mathcal{D}_{U}\right\Vert _{1}=O\left(n^{2}\epsilon\log m\right).\]

\begin{cor}
\label{cor:Component accuracy}In order to have $\left\Vert \mathcal{D}_{\tilde{U}}-\mathcal{D}_{U}\right\Vert _{1}=o\left(1\right)$,
it suffices for every component in the network to have an accuracy
of $\left\Vert \tilde{A}-A\right\Vert _{\mathrm{op}}=o\left(\frac{1}{n^{2}\log m}\right)$.
\end{cor}

\section{Interpretation}

Note that we do not obtain that a constant error suffices. In fact,
constant error does not suffice, as shown in \cite{Kalai Kindler,Leverrier Patron},
suggesting that fault-tolerance is necessary to perform scalable quantum
computing. This is not surprising -- we expect that more photons require
higher accuracy for the unitary because each photon interaction with
the unitary introduces error. Similarly, as the network requires more
and more components, each component must have better accuracy to maintain
the same overall accuracy. 

We conjecture that the requirement we obtain that $\left\Vert \tilde{U}-U\right\Vert _{op}=o\left(\frac{1}{n}\right)$
is the best possible. It parallels the Berstein-Vazirani result for
qubit-based circuits \cite{BV}. Because each photon passes through
the network and experiences its imperfections, it is natural to conjecture
that the acceptable error in the network falls inversely with the
number of photons. Likewise, since each photon passes through a depth
$O\left(n\log m\right)$-network, one might has guessed that the acceptable
error of each component is $O\left(\frac{1}{n\log m}\right)$ time
that of the full network, as corresponding to the sufficient bound
in \ref{cor:Component accuracy}.

\subsection{Future work}

This work solely addresses one type of noise: errors in the beamsplitters
and phaseshifters that cause them to implement a slightly erronnous
unitary matrix. We would like to extend these results to other sources
of noise. The more plausible potential extensions of this approach
are those dealing with continuous errors rather than discrete ones
like photon losses. One such source is the partial distinguishability
of the photons as they pass through the network, a phenomenon that
has been mathematically modelled by Tichy \cite{Tichy} and Xu \cite{Xu}.

The gaps between the sufficient bound proven here and the necessarry
bounds proven in \cite{Leverrier Patron,Kalai Kindler} mean that
an improvement must be possible to at least one of the sides. Moreover,
all the results are fine-tuned for models of noise, so it would be
ideal to bound the error under each of the noise models.

\section{appendix A: Translation Noise Models\label{sec:appendix-A}}

Previous work on BosonSampling noise \cite{Leverrier Patron,Kalai Kindler}
used different measures of error than we did. In order to put these
results on the same scale as ours, we will find the amount of operator
distance error that corresponds to the errors they prove prohibitive.
Note that because these results are optimized for their specific model
of error, the converted results are not necessarily the strongest
possible.

The work of Leverrier and Patrón \cite{Leverrier Patron} demonstrates
that each linear optical element must have fidelity $1-O\left(1/n^{2}\right)$.
This corresponds to operator distance $O\left(1/n^{2}\right)$ for
each element. From the observation in Section \ref{sub:Components}
that the operator distance of the whole network is at most it depth
times that of each component, and the result that $O\left(n\log m\right)$
depth suffices (Theorem 45 of \cite{key-1}), this corresponds to
necessary error $O\left(\log m/n\right)$.

The work of Kalai and Kindler \cite{Kalai Kindler} argues that a
noise level of additive $\epsilon=\omega\left(1/n\right)$ Gaussian
error is prohobitive for Boson Sampling. We show that this corresponds
to operator distance \[
\left\Vert \tilde{U}-U\right\Vert _{op}=\omega(1/\sqrt{n})\]
so that we may put it on the same scale as out result.

Consider an $\epsilon$-noise of a matrix $X$. In order to match
with operator distance, we consider $X$ to be the entire $m\times m$
unitary matrix, rather than an $n\times n$ submatrix, since we expect
the error to affect entries in the whole matrix just as it does the
submatrix. Since each entry of a unitary matrix has a norm of $1/\sqrt{m}$
in RMS average, the error should be $\epsilon/\sqrt{m}$.

So, an $\epsilon$-noise of a unitary matrix $U$ is given by 

\[
\tilde{U}=\sqrt{1-\epsilon}U+\sqrt{\epsilon}G/\sqrt{m},\]
where $G$ is a matrix of i.i.d. complex Gaussians. To first order
in $\epsilon$, the difference $\tilde{U}-U$ is given by 

\[
\tilde{U}-U=-\epsilon U/2+\sqrt{\epsilon}G/\sqrt{m}+O(\epsilon^{2})\]
Since $U$ and $G/\sqrt{m}$ have entries of the same RMS-norm, for
small $\epsilon,$ the term with coefficient $\sqrt{\epsilon}$ dominates
the remaining terms:

\[
\tilde{U}-U=\sqrt{\epsilon}G/\sqrt{m}+O(\epsilon)\]
Then, the phohibitive amount of noise $\epsilon=\omega(1/n)$ corresponds
to 

\[
\tilde{U}-U=\omega(1/\sqrt{m})G/\sqrt{m}\]
Finally, with the result from \cite{Random norms} that a random $m\times m$
Gaussian matrix has operator norm $\Theta\left(\sqrt{m}\right)$ with
high probability, $\left\Vert G/\sqrt{m}\right\Vert _{op}=\Theta\left(1\right)$,
and so the corresponding operator distance is \[
\left\Vert \tilde{U}-U\right\Vert _{op}=\omega(1/\sqrt{n})\]

\end{document}